\documentclass[twocolumn, aps, pra, showpacs, superscriptaddress, floatfix, nofootinbib, 10pt]{revtex4-1}

\usepackage{graphicx}
\usepackage{amsmath,amsthm,amsfonts}
\usepackage{tabularx}
\usepackage{array}
\usepackage{bbm}
\usepackage{bm}
\usepackage[T1]{fontenc}
\usepackage[colorlinks]{hyperref}
\usepackage{cleveref}

\newtheorem{theorem}{Theorem}
\newtheorem{corollary}{Corollary}
\newtheorem{lemma}[theorem]{Lemma}

\begin{document}

\title{A Logical Proof of Quantum Correlations Requiring Entanglement Measurements}

\author{Adel Sohbi}
\email{sohbi@kias.re.kr}
\author{Jaewan Kim}
\affiliation{School of Computational Sciences, Korea Institute for Advanced Study, Seoul 02455, Korea}

\begin{abstract}
We present a logical type of proof of contextuality for a two-qubit state. We formulate a paradox that cannot be verified by a two-qubit system with local measurements while it is possible by using entanglement measurements. With our scheme we achieve $p_{\rm Hardy} \approx 0.167$, which is the highest probability obtained for a system of similar dimension. Our approach uses graph theory and the global exclusivity principle to give an interpretation of logical type of proofs of quantum correlations. We review the Hardy paradox and find connection to the KCBS inequality. We apply the same method to build a paradox based the CHSH inequality.
\end{abstract}

\date{\today}
\maketitle


Contextuality refers to the property of measurements statistics to have their outcomes depending on the set of measurements that are simultaneously performed. Originally discovered by Bell, Kochen and Specker \cite{Bell:rmp66,KS:jmm67}, it uses the notion of compatibility, which refers to the property of several measurements to be performed simultaneously. Recent efforts have been made on understanding contextuality using a graph theory approach \cite{AB:njp11,FLS:qpl13,CSW:arXiv10,CSW:prl14}. Contextuality has shown an advantage in quantum computing \cite{AB09,raussendorf2013contextuality,HWVE:nat14,FRB:arXiv18}, in quantum cryptography \cite{SBKTP09,HHHHPB10}, quantum communication \cite{SHP:arXiv17,SC:arXiv18} and state discrimination \cite{SS:prx18}.

Contextuality generalizes non-locality \cite{FLS:qpl13} and the ability to experimentally observe these two fundamental properties of quantum physics is an important point when it comes to their uses in information technology tasks. It is then one of the most fundamental challenge to develop tools to characterize and classify measurements statistics belonging to different theories. A common technic to witness non-locality and contextuality is through inequalities of measurements statistics. Another possibility is by using logical arguments \cite{GHZ:89,GHZ:aps90,Hardy:prl93,Mermin:rmp93,WM:prl12}, where a set of conditions if satisfied assert the non-local or contextual nature of the system considered. While it is clear that a verification of non-locality or contextuality with a logical proof implies the violation of an inequality \cite{BCPSW:rmp14}, the inverse is not straightforward to answer.

When such properties are observed on spatially separated systems, entanglement is a needed resource. Entangled states and local measurements are a sufficient resource to observe non-locality. Moreover, a quantum measurement can also have entanglement properties when its eigenstates are entangled. This is a crucial key in the architectures of quantum network using quantum repeaters \cite{SSRG:rmp11}, where independent sources send entangled photons to distant parties. When entanglement measurements are performed on pair of photons from different sources it can create entanglement between initially uncorrelated parties. In particular Bell state measurements can create maximally entangled states and can also be self-tested device-independently \cite{RKB:arXiv18,BSS:arXiv18}.

Recent advances on the classification of correlations has been made with a graph theoretical approach \cite{CSW:arXiv10,FLS:qpl13,CSW:prl14}. A proof of quantum correlations can directly be made by using an inequality which is derived by the graphs properties. logical proofs of quantum correlations within the graph theory approach has been studied \cite{CBCB:prl13,SZDM:pra16}, but not yet widely known.

Contextuality relies on the global exclusivity principle \cite{Cabello:prl13}: the sum of probabilities of pairwise exclusive events cannot exceed 1. The exclusivity principle is a fundamental principle under which quantum mechanics stands. In particular, it provides the bounds under which quantum correlations exist \cite{Cabello:prl13,Yan:pra13,Cabello:pra14,BCC:pra14}. More recently, it has been proved that monogamy arises as a consequence of the exclusivity principle \cite{JCWGC:arXiv14}.

By using the global exclusivity principle, we develop a method to build a graph construction for logical proofs of non-locality and contextuality. After revisiting the Hardy paradox \cite{Hardy:prl93} to exhibit the graphic properties of logical proofs, we show how our technique can apply to generate a paradox from the CHSH inequality \cite{CHSH:prl69} in a similar way that has been done for the KCBS inequality \cite{KCBS:prl08,CBCB:prl13}. We show that this paradox cannot be verified by a two-qubit system with local measurements while it is possible with entanglement measurements. We finally provide a contextuality inequality which is violated when the paradox is verified.

This letter is structured as follows. We review the approach in \cite{CSW:prl14} to introduce how a graph can be used to represent measurement outcomes obtained in an experiment. Then, we revisit the Hardy paradox \cite{Hardy:prl93} for two-qubit and give an interpretation of this paradox in the graph theory approach by using the exclusivity principle. We apply our method to the CHSH inequality \cite{CHSH:prl69} in order to build a paradox. Finally, we compare three different resources to demonstrate this paradox: classical theory, two-qubit state with local measurements and with entanglement measurements.

\emph{Exclusivity Graph.---}In a given experiment, we can represent on a graph each possible event $e$ corresponding to the tuples of outcomes of a given set of compatible measurements. For instance if two dichotomic measurements (``0'' and ``1'' outcomes) denoted by $x$ and $y$ are compatible we have the following list of possible events: $(0,0\vert x,y),(0,1\vert x,y),(1,0\vert x,y),(1,1\vert x,y)$.
In the bipartite scenario, we define a graph $\mathcal{G}(V,E)$ in which we associate to each vertex an event $e$ denoted by $(a,b\vert x,y)$, where $a$ and $b$ are the measurement outcomes for tests $x$ and $y$ in an experiment. Pairs of exclusive events are represented by adjacent vertices. Two events are exclusive if they involve at least one same measurement with non consistent outcomes. For instance, we consider the two following events: $(0,0\vert x,y)$ and $(1,0\vert x,y')$, with $y \neq y'$, in this example the outcomes of the measurement $x$ are different in the two events, obviously, this cannot occur simultaneously. Hence, this two events are called exclusive. $\mathcal{G}$ is defined as the \textit{exclusivity graph of the experiment} \cite{CSW:prl14} . In quantum theory, exclusive events are represented by orthogonal projective measurements (see Lemma~\ref{lem:LO}). Contextuality inequalities can be build from the graph properties as explained in Appendix~\ref{sec:CSW14}.

\emph{Hardy Paradox.---}The Hardy Paradox \cite{Hardy:prl93} is a set of four probabilistic conditions in the bipartite scenario, where Alice and Bob have two dichotomic measurements (denoted by 0 and 1).
\begin{align}
& P(00|00)>0,\label{eq:HP1}\\
& P(00|01)=0,\label{eq:HP2}\\
& P(00|10)=0,\label{eq:HP3}\\
& P(11|11)=0,\label{eq:HP4}
\end{align}
where an event $(a,b\vert x,y)$ means that Alice uses her measurement $x \in \{0,1\}$ and gets the outcome $a \in \{0,1\}$ and Bod uses his measurement $y \in \{0,1\}$ and gets the outcome $b \in \{0,1\}$. We refer to this scenario as 2-2-2 scenario (two parties, two measurements with two outcomes).
It can be shown that the conditions in \cref{eq:HP1,eq:HP2,eq:HP3,eq:HP4} are not consistent with any Local Hidden Variable theory (LHV), however it is possible to find a quantum system and a set of measurements satisfying the above conditions \cite{Hardy:prl93}.

We aim to revisit the Hardy paradox with a graph theory approach, to do so, we consider the exclusivity graph of the experiment. This scenario (bipartite with two dichotomic measurements) is well known and the exclusivity graph of the experiment is already given in \cite{CSW:prl14}. From the exclusivity principle, we can write:
\begin{align}
& P(00|01) + P(10|01) + P(01|11) + P(11|11) = 1,\label{eq:HP_EP1}\\
& P(01|10) + P(00|10) + P(11|11) + P(10|11) = 1.\label{eq:HP_EP2}
\end{align}
In quantum theory, the two above equations can be verified by using the completeness properties of the different measurements after using the Born rule. More generally, to illustrate this saturation of the exclusivity principle, we represent in Tab.~\ref{tab:SEPHP} (See Appendix~\ref{sec:SEPHP}) an exhaustive list of all possible events for this scenario. We can see that all the four events in each equations in \cref{eq:HP_EP1,eq:HP_EP2} cover all the 16 events of the CHSH scenario. Hence, their probabilities sum to one.

Because of the conditions in \cref{eq:HP2,eq:HP3,eq:HP4} \cref{eq:HP_EP1,eq:HP_EP2} can be simplified to:
\begin{align}
& P(10|01) + P(01|11) = 1,\label{eq:HP_EP_sat1}\\
& P(01|10) + P(10|11) = 1.\label{eq:HP_EP_sat2}
\end{align}
In a deterministic model, i.e., a model where each event $e_i$ has a probability $P(e_i)$ equals to $0$ or $1$, the \cref{eq:HP_EP_sat1,eq:HP_EP_sat2} imposes that $P(10|01)=1$ or $P(01|11) = 1$ and $P(01|10)=1$ or $P(10|11) = 1$. Under such model we have $P(00|00)=0$ and contradicts with \cref{eq:HP1} while this quantity could be non-zero by using a quantum resource \cite{Hardy:prl93}.

\begin{figure}[ht]
      \centering
      \includegraphics[width=6cm]{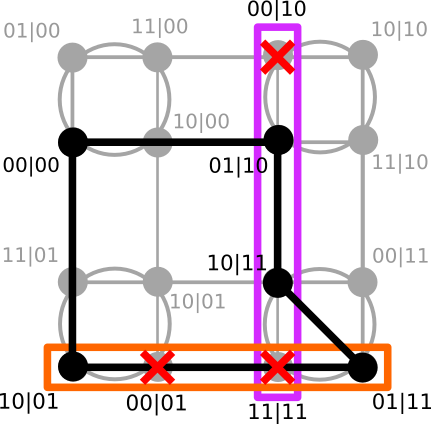}
      \caption{Exclusivity graph of the experiment for the 2-2-2 scenario \cite{CSW:prl14}. All events except the one appearing in \cref{eq:HP_EP_sat1,eq:HP_EP_sat2} and the event $(00|00)$ are in grey color. We can see that the subgraph formed by the remaining vertices is a pentagon. The two rectangles (orange and purple) represent the set of events appearing in the \cref{eq:HP_EP1,eq:HP_EP2}. The 'X' (red) represents the events with probability null imposed in the formulation of the paradox in \cref{eq:HP2,eq:HP3,eq:HP4}.}
      \label{fig:HP}
\end{figure}

In Fig.~\ref{fig:HP} are represented all of the 16 events of the 2-2-2 scenario. As in \cite{CSW:prl14}, for simplicity pairwise exclusive events are represented on straight lines or circles. Hence, events on a same line or circle follow the exclusivity principle, the sum of their probability is upper-bounded by 1. We have kept all the 16 events, however, all of the unwanted ones are in light grey color. The subgraph formed by the five events $\{(10|01),(01|11),(01|10),(10|11),(00|00)\}$ form a pentagon on the exclusivity graph of the experiment.

A deterministic model can be seen as a coloring problem where assigning a probability equal to 1 to two exclusive events is not possible, while other combinations are possible. A probability equal to 1 means that an event always occurs and two exclusive events cannot happen simultaneously. It is not possible to assign a probability 1 to the event $(00 \vert 00)$ while having an assignment satisfying the condition in \cref{eq:HP_EP_sat1,eq:HP_EP_sat2}. The five events form a pentagon which is the simplest non-perfect graph \cite{CRST:am06}, which is known to be a resource for building contextuality inequalities \cite{CSW:prl14}.

We recover the logical proof of contextuality developed in \cite{CBCB:prl13} which is based on the the KCBS inequality scenario \cite{KCBS:prl08}, $S_{KCBS} = \sum_{i=1}^5 P(e_i) \leq 2$. In this case the maximum violation of the KCSB inequality under the paradox conditions is given by $2+1/9 \approx 2.11$ \cite{CBCB:prl13}. When we use a separated system of two qubits (Hardy's paradox), this quantity is about $2.09$. This gap has also been observed for the violation of the KCBS without considering any paradox. In this case, the maximum violation of the KCBS inequality is $2.236$ for a qutrit system \cite{KCBS:prl08} whereas for a bipartite system it goes up to $2.207$ \cite{SBBC:pra13}.

\emph{Graphical construction of the CHSH Paradox.---}The CHSH inequality can be written with the probabilities of the events from its exclusivity graph $G_{\rm CHSH} = (V_{CHSH,}E_{\rm CHSH})$ as follows \cite{CSW:prl14},

\begin{equation}\label{eq:schsh}
  S_{\rm CHSH} = \sum_{i\in V_{\rm CHSH}} P(e_i) \leq 3.
\end{equation}

We now show how to build a paradox in a similar way as previously done for the Hardy paradox.
\begin{figure}[ht]
      \centering
      \includegraphics[width=6cm]{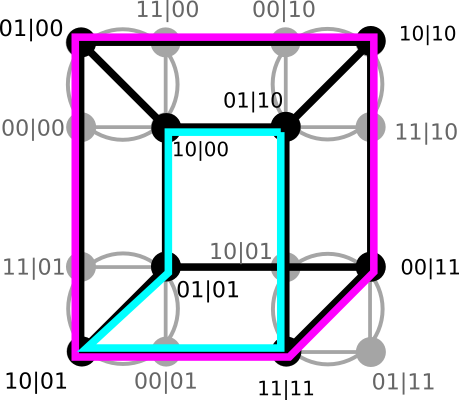}
      \caption{Exclusivity graph of the CHSH inequality in the form given in \cite{CSW:prl14}. The two pentagons are shown with two different colors.}
      \label{fig:CHSH_graph}
\end{figure}

 In Fig.\ref{fig:CHSH_graph}, is shown the exclusivity graph of the CHSH inequality. One can notice that by using two overlapping pentagons we can cover all the set of events in $G_{\rm CHSH}$. Individually, each pentagon forms a Hardy paradox and by assembling them one can ask whether a paradox which uses the events in the CHSH inequality can be formed. In order to do that we first need to formulate such  paradox by using the saturation of the exclusivity principle. Because the two pentagons are overlapping they share one exclusivity principle equation.

To illustrate the saturation of the exclusivity principle, we represent in Tab.~\ref{tab:SEPCHSH} (See Appendix~\ref{sec:SEPCHSH}) an exhaustive list of events for this scenario. The saturation of the exclusivity principle from the chosen pentagons in Fig.~\ref{fig:CHSH_graph} goes as follow:
\begin{align}
&P(10|00) + P(01|01) = 1,\label{eq:CHSH_EP_sat1}\\
&P(10|01) + P(11|11) = 1,\label{eq:CHSH_EP_sat2}\\
&P(01|10) + P(00|11) = 1.\label{eq:CHSH_EP_sat3}
\end{align}
\begin{figure}[ht]
      \centering
      \includegraphics[width=6cm]{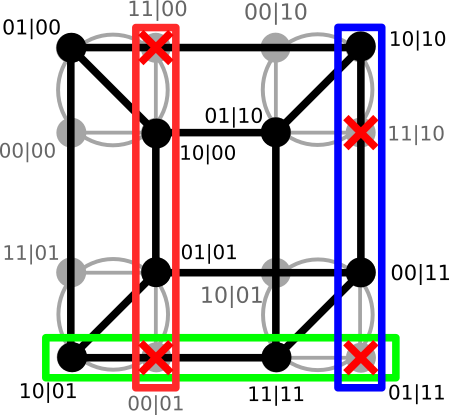}
      \caption{Exclusivity graph of the CHSH inequality in the form given in \cite{CSW:prl14}. All events except the one appearing in the inequality \cref{eq:schsh} are in grey color. The three rectangles (red, blue and green) represent the sets of events appearing in the \cref{eq:CHSH_EP_sat1,eq:CHSH_EP_sat2,eq:CHSH_EP_sat3}. The 'X' (red) represents the events with probability null to have the saturation of the exclusivity principle.}
      \label{fig:CHSH_exclu}
\end{figure}
The formulation of such assemblage of paradoxes is the following. We call the CHSH paradox the following sets of conditions in the 2-2-2 scenario:
\begin{align}
& P(01|00)+P(01|10)>0,\label{eq:CHSH_paradox1}\\
& P(11|00)=0,\label{eq:CHSH_paradox2}\\
& P(00|01)=0,\label{eq:CHSH_paradox3}\\
& P(11|10)=0,\label{eq:CHSH_paradox4}\\
& P(01|11)=0.\label{eq:CHSH_paradox5}
\end{align}
Interestingly, if the CHSH paradox is verified then the CHSH inequality in the form in \cref{eq:schsh} must be violated as well. Indeed if a physical system verifies the equations in \cref{eq:CHSH_EP_sat1,eq:CHSH_EP_sat2,eq:CHSH_EP_sat3} then $S_{\rm CHSH} = 3 + P(01|00) + P(01|10)$ which must be greater than $3$ because $P(01|00)+P(01|10)>0$.

We are now interested to know whether the conditions in \cref{eq:CHSH_paradox1,eq:CHSH_paradox2,eq:CHSH_paradox3,eq:CHSH_paradox4,eq:CHSH_paradox5} form a paradox that witnesses non-locality and contextuality. To do so, we can now verify which system has the possibility to verify all the conditions in \cref{eq:CHSH_paradox1,eq:CHSH_paradox2,eq:CHSH_paradox3,eq:CHSH_paradox4,eq:CHSH_paradox5}.

\emph{CHSH Paradox with classical resources.---}The CHSH paradox is the assemblage of two different Hardy paradoxes, \emph{i.e.}, two paradoxes built on pentagons. We already know that each of these paradoxes cannot be verified individually by a system obeying any classical theory and by extension the accumulation of the two neither.

\begin{theorem}\label{thm:CHSH_NCHV}
A Non-Contextual Hidden Variable theory (NCHV) cannot verify the CHSH Paradox.
\end{theorem}

In other words, it is impossible to have a probability equal to one by deterministic theory to the vertices associated to the events $(01|00)$ and $(01|10)$ in a consistent way with the conditions imposed with \cref{eq:CHSH_EP_sat1,eq:CHSH_EP_sat2,eq:CHSH_EP_sat3}. See Append.~\ref{sec:CHSHC} for a detailed proof.

\emph{CHSH Paradox with a two-qubit state and local measurements.---}We know that each paradox (one for each pentagon in Fig.\ref{fig:CHSH_graph}) can be individually addressed by a two-qubit state with local measurements. However, as stated in the following theorem it is no longer the case for the combination of the two pentagons as formed in the CHSH paradox in \cref{eq:CHSH_paradox1,eq:CHSH_paradox2,eq:CHSH_paradox3,eq:CHSH_paradox4,eq:CHSH_paradox5} .

\begin{theorem}\label{thm:CHSH_QLM}
Given two parties, Alice and Bob, sharing a two-qubit entangled state $\vert \psi \rangle_{AB}$. There is no possible set of local measurements $A_0$ and $A_1$ (resp. $B_0$ and $B_1$) for Alice's dichotomic measurements on her qubit (resp. Bob) that verifies the CHSH paradox.
\end{theorem}

A detailed proof is given in Append.~\ref{sec:CHSHQ2D2D}. This impossibility to verify the CHSH paradox comes from the necessity to address both pentagon simultaneously. For instance the conditions $P(11|00)=0$ and $P(00|01)=0$ if verified simultaneously would not permit the use of an entangled state or impose compatibility of the measurements of one of the party. Entanglement and incompatibility of the measurements a necessary to observe non-locality.

\emph{CHSH Paradox with a two-qubit and entanglement measurements.---}While an entangled state can take any form in the Hilbert space, in a separated system the operations are restricted to local operations only. In quantum mechanics, the spatial separation implies that Alice's measurements (denoted by observables $A_0$ and $A_1$) and Bob's quantum measurements (denoted by observables $B_0$ and $B_1$) act in different subspaces and hence always commute to each other: $[A_i\otimes \mathbb{I},\mathbb{I} \otimes  B_j]=0$, $\forall (i,j)\in \{1,2\}$. For dichotomic measurements ($\pm1$ outcome value), the observable can be defined by two rank-1 projectors, for instance, $A_i = \Pi_{A_i,+} - \Pi_{A_i,-}$, with $ \Pi_{A_i,+} + \Pi_{A_i,-} = \mathbb{I}$. The observable acting on the whole quantum system is given by the tensor product of the observable of Alice and Bob. As a consequence, the eigenvectors of the observable are also given by the tensor product of two vectors from each subspace.

With entanglement measurements, instead of effectively acting on the local systems separately, observables are now acting on the whole space freely. However, the same commutation relations need to be verified: $[A_i,B_j]=0$. Forming the set of four observables following the graph conditions and the proper commutations relations with entanglement measurement is actually not a trivial problem. However, we can define another scenario with the same graph of experiment by following the graph formalism in \cite{CSW:arXiv10} where each projector associated to a vertex defines an observable as $A_i$ with $\pm 1$ outcome value associated to the vertex with $i$ (see Fig.~\ref{fig:CHSH_contineq} in Append.~\ref{sec:CHSHCSW10}). We can construct a contextuality inequality.

\begin{theorem}\label{thm:CHSH_QEntM}
Given the exclusivity graph of the CHSH, $G_{\rm CHSH}(V_{\rm CHSH},E_{\rm CHSH})$ where a set of dichotomic measurements $\{A_i\}$ with $\pm 1$ outcomes values assigned to the vertices $\{i\}$ such that $\forall (i,j) \in E_{\rm CHSH}$ $A_i$ and $A_j$ are compatible. Then all Non-Contextual Hidden Variable (NCHV) theory are bounded by the following inequality
\begin{align}\label{eq:newcontineq}
\sum_{(i,j)\in E_{\rm CHSH}} \langle A_i A_j \rangle  \underset{{\rm NCHV}}{\geq} -6,
\end{align}
\end{theorem}

In the Append.~\ref{sec:CHSHCSW10} we provide the proof where we show that this inequality has the same classical and quantum bound compared to the CHSH.

In this case, the \cref{eq:CHSH_EP_sat1,eq:CHSH_EP_sat2,eq:CHSH_EP_sat3} means that the state $\vert \psi \rangle$ has to be generated by the pairs of vectors ($\vert v_2 \rangle$, $\vert v_3 \rangle$), ($\vert v_4 \rangle$, $\vert v_5 \rangle$) and ($\vert v_6 \rangle$, $\vert v_7 \rangle$) at the same time. In \cref{eq:CHSH_entmeas1,eq:CHSH_entmeas2,eq:CHSH_entmeas3,eq:CHSH_entmeas4,eq:CHSH_entmeas5,eq:CHSH_entmeas6,eq:CHSH_entmeas7,eq:CHSH_entmeas8,eq:CHSH_entmeas9} we present an example of a set of state and measurements that verify the orthogonality relations from the graphs and the conditions from \cref{eq:CHSH_paradox1,eq:CHSH_paradox2,eq:CHSH_paradox3,eq:CHSH_paradox4,eq:CHSH_paradox5}.
\begin{align}
& \vert \psi \rangle = (1,1,0,0)^T/\sqrt{2},\label{eq:CHSH_entmeas1}\\
& \vert v_1 \rangle =  (0,-1,-2,1)^T/\sqrt{6},\label{eq:CHSH_entmeas2}\\
& \vert v_2 \rangle =  (1,0,0,0)^T,\label{eq:CHSH_entmeas3}\\
& \vert v_3 \rangle =  (0,1,0,0)^T,\label{eq:CHSH_entmeas4}\\
& \vert v_4 \rangle =  (0,2,1,1)^T/\sqrt{6},\label{eq:CHSH_entmeas5}\\
& \vert v_5 \rangle =  (3,1,-1,-1)^T/2\sqrt{3},\label{eq:CHSH_entmeas6}\\
& \vert v_6 \rangle =  (2,0,1,-1)^T/\sqrt{6},\label{eq:CHSH_entmeas7}\\
& \vert v_7 \rangle =  (1,3,-1,1)^T/2\sqrt{3},\label{eq:CHSH_entmeas8}\\
& \vert v_8 \rangle =  (1,0,2,1)^T/\sqrt{6},\label{eq:CHSH_entmeas9}
\end{align}
where $T$ is the transpose operator. This leads to $P(01|00)+P(01|10)= 1/6$. Moreover, we have
\begin{align}\label{eq:CHSH_ququart_gen}
& \vert \psi \rangle = (\vert v_2 \rangle + \vert v_3  \rangle)/\sqrt{2},\\
& \vert \psi \rangle = (\vert v_4 \rangle + \sqrt{2}\vert v_5 \rangle)/\sqrt{3},\\
& \vert \psi \rangle = (\vert v_6 \rangle + \sqrt{2}\vert v_7 \rangle)/\sqrt{3},
\end{align}
which corresponds to the verification of the \cref{eq:CHSH_EP_sat1,eq:CHSH_EP_sat2,eq:CHSH_EP_sat3}. We conclude that a two-qubit state with entanglement measurements can exhibit quantum correlations with the CHSH paradox. The inequality in \cref{eq:newcontineq} is then violated. The observables can be obtained from the projectors (see App.~\ref{sec:CHSHCSW10}).

We use $p_{\rm Hardy}$ to refer to $P(01|00)+P(01|10)$ for the 2-2-2 scenario and $P(1|1)+P(1|8)$ in the contextual scenario. We have summarized the different values of $p_{\rm Hardy}$ from the different theories in Tab.\ref{tab:hierarchy}. Interestingly, a two-qubit system with local measurements does not perform better than classical theory, while a two-qubit system with entanglement measurements can verify the paradox.
\begin{table}[ht]
\begin{tabular}{l|c|c|c|}
     & Classical Th. &  Local M. & Entangled M. \\
    \hline
		\hline
    $p_{\rm Hardy}$ & 0 & 0 & 1/6
\end{tabular}
		\caption{Summary of the hierarchy of the verification of the paradox for different resources. "Classical Th." is a classical resource,  "Local M." is a two-qubits state with local measurements and "Entangled M." with entanglement measurements.}
		\label{tab:hierarchy}
\end{table}
Using the same graph of experiment with different events assigned from with different scenario has already been investigated for the pentagon. For the KCBS inequality, there is also a hierarchy \cite{KCBS:prl08,SBBC:pra13} between a two-qubit state with local measurements and a single qutrit system, \emph{i.e.}, a qutrit can get a violation higher than what is possible to get for a two-qubit state. However, the differences here is that a two-qubit state does not outperform classical theories. Because paradox proofs are difficult to verify experimentally, an interesting open problem is to find an inequality with similar features.
\begin{table}[ht]
\begin{tabular}{l|c|c|c|}
     & $p_{\rm Hardy}$ &  dim & \#M \\
    \hline
		\hline
    CHSH Paradox (our) & 1/6  $\approx$ 0.167 & 4 & 8  \\
    \hline
    Cabello \textit{et al.} \cite{CBCB:prl13} & 1/9  $\approx$ 0.11 & 3 & 5 \\
    \hline
    Hardy Paradox \cite{Hardy:prl93} & $\frac{5\sqrt{5}-11}{2}$  $\approx$ 0.09 & 4 & 4\\
\end{tabular}
		\caption{Comparison of the different paradox proof of quantum correlations with the CHSH paradox and for other system of similar dimension ("dim" is the dimension of the required Hilbert space and \#M is the number of measurements) from the literature.}
		\label{tab:phardy}
\end{table}

In the Tab.~\ref{tab:phardy} we compare different maximum values of $p_{\rm Hardy}$ paradox for system of similar dimension. Our formulation of a paradox obtain the highest value of $p_{\rm Hardy}$ among the different the previous logical proofs using system of similar dimension. However, compared to the other approaches, the fundamental difference here is that our $p_{\rm Hardy}$ is the sum of two probabilities.

\emph{Conclusions.---}We have described a method based on the saturation of the exclusivity principle to build a graph for a logical proof of quantum correlations. Our method suggests a deep connection between the exclusivity principle and the logical proofs of quantum correlations. With this method we have reviewed the Hardy paradox and have showed direct connection with the KCBS scenario by showing that the graph representation of the Hardy paradox is a pentagon. We demonstrate the composition of two paradoxes and archive the highest $p_{\rm Hardy}$ for system of similar dimension. The verification of the paradox implies the violation if the CHSH inequality in the probability form, but the inverse is not true. Interestingly, we found that this paradox cannot be verified by a two-qubit state with local measurements while it is possible with entanglement measurements. Finally, we also provided a contextuality inequality which is violated when the paradox is verified.
There exist a multipartite extension of the Hardy paradox in \cite{WM:prl12}, it would be interesting to know whether the superposition of paradoxes occur also in the multipartite settings.

\noindent {\bf Acknowledgements.}
The authors would like to thank Shane Mansfield, Damian Markham and Wonmin Son for discussions.

\bibliographystyle{unsrt}

\bibliography{biblio}

\clearpage
\appendix

\section{Graph-Theoretic Approach to Quantum Correlations}\label{sec:CSW14}

\subsection{Graph Formalism}

We propose to review the graph formalism for the description of quantum correlations developed in \cite{CSW:prl14}. In a scenario, we can list each possible event $e$ for a given set of compatible measurements. For instance if two dichotomic measurements denoted by $x$ and $y$ are compatible we have the following list of possible events: $(0,0\vert x,y),(0,1\vert x,y),(1,0\vert x,y),(1,1\vert x,y)$. From that, we can define a graph $\mathcal{G}(V,E)$ for which we associate to each vertex an event $e$ denoted by $(a,b\vert x,y)$, where $a$ and $b$ are the measurement outcomes for tests $x$ and $y$. Pairs of exclusive events are represented by adjacent vertices. Two events are exclusive if they involved at least one measurement with two non consistent outcomes. For instance, if we consider the two following events: $\{0,0\vert x,y\}$ and $\{1,0\vert x,z\}$, the outcomes of the measurement $x$ are different in the two events, obviously, this cannot occur simultaneously. Hence, this two events are called exclusive. We will refer to $\mathcal{G}$ as the \textit{exclusivity graph of the experiment} \cite{CSW:prl14}.

We are also interested in the measurement outcomes statistics. In particular, contextuality and non-locality inequalities are based on linear combination of probabilities of a subset of events of the corresponding experiments. A specific combination, $S$ will be of the form: $S=\sum_i w_i P(e_i)$, where $w_i > 0$ and $P(e_i)$ is the probability of the event $e_i$ to happen. We define a graph $G(V,E,w)\subseteq \mathcal{G}$, where each vertex $i\in V$ represents an event $e_i$ of $S$. We call $G$ the exclusivity graph of $S$ \cite{CSW:prl14}.

\subsection{Contextuality Inequality}
The exclusivity graphs can be used to calculate limits of the correlations in classical and quantum theories. It has been showed that \cite{CSW:prl14}:
\begin{equation}
S \leq \alpha_w(G)\leq \vartheta_w(G),
\end{equation}
where $\alpha(G)$ is the independent number of the graph $G$ and $\vartheta(G)$ is the Lovasz number of the graph $G$. The equality of $\alpha_w(G)$ and $\vartheta_w(G)$ holds if $G$ is a perfect graph, i.e., the graph $G$ has neither odd-holes of length at least five nor odd anti-holes \cite{CRST:am06}.

\subsubsection{Classical Theory}
The maximum value of $S$ with a non-contextual classical resource is always upper-bounded by a non-contextual deterministic behavior, i.e. a probability distribution of the outcomes where each event $e_i$ in $S$ has a probability $P(e_i)$ equals to $0$ or $1$. Exclusive events cannot have probability $1$ simultaneously. Hence, the maximum value for $S$ is given by the maximum number of events with a probability equal to $1$. This correspond to the maximum number or $1$ we can assign to non-adjacent vertices in $G$, which correspond to the independence number $\alpha(G)$ of the graph.

\subsubsection{Quantum Theory}
In a quantum theory, the set of measurements that leads to the events $\{e_i\}_{i\in V}$ can be represented by a set of projectors $\{\Pi_i\}_{i\in V}$. The maximum value of $S$ in this case is $max \sum_{i\in V} Tr[\Pi_i \rho]$, where $\rho$ is a quantum state and simplified to $max \sum_{i \in V} \vert \langle v_i \vert \psi \rangle \vert^2$, where $\vert v_i \rangle = \Pi_i \vert \psi \rangle / \sqrt{\langle \psi \vert \Pi_i \vert \psi \rangle}$ and $\vert \psi \rangle$ is a pure quantum state. This quantity is known to be equal to the Lovasz number $\vartheta(G)$, of the graph $G(V,E,w)$ where the set of vectors $\{\vert v_i \rangle\}_{i \in V}$ is the orthonormal representation (OR) of the complement $\bar{G}$ and the quantum state $\vert \psi \rangle$ is the handle \cite{Lovasz:79}.

\section{Saturation of the Exclusivity Principle in the 2-2-2 Scenario}\label{sec:SEP}
Let be $E= \{e_i\}$, a set of mutually exclusive events. The exclusivity principle goes as follows:
\begin{equation}
 \sum_{i} P(e_i) \leq 1,
\end{equation}
the sum of the probability of mutually exclusive events is bounded by 1. In particular, the normalization condition is the saturation of this principle when all events arise from the same measurement choice.

In the 2-2-2 scenario (2 parties, 2 measurements each with 2 outcomes), four mutually exclusive events saturate the exclusivity principle.
\begin{equation}
 \sum_{a,b\in\{0,1\}} P(a,b\vert x,y) = 1, \forall x,y \in \{0,1\}.
\end{equation}

In quantum mechanics, events are associated with projectors and related with the Born rule: $P(e) = Tr(\rho \Pi_e)$, where $\rho$ is a quantum state.

\begin{lemma}\label{lem:LO}
Given two events $e_1$ and $e_2$ and their associated projectors $\Pi_1$ and $\Pi_2$. If the events are exclusive, then the two associated projectors are orthogonal.
\end{lemma}

This property can be rephrased and referred as a feature called Local Orthogonality (LO) in the literature where exclusive events are referred as orthogonal events \cite{SFABCLA:pra14}.

\begin{theorem}
Given a set of mutually exclusive events $E= \{e_i\}$. If their associated projectors verify the completeness relation, then the sum of the probabilities of the events in $E$ saturates the exclusivity principle.
\end{theorem}

\begin{proof}
Let a set of mutually exclusive $E= \{e_i\}$ and their associated projectors $\{\Pi_{e_i}\}$ such that $\sum_{e_i \in E} \Pi_{e_i} = \mathbb{I}$. Then for any quantum state $\rho$:
\begin{equation}
 \sum_{i} P(e_i) = Tr(\rho \sum_{e_i \in E} \Pi_{e_i} ) = Tr(\rho) = 1.
\end{equation}
\end{proof}
While for our purpose we only need to consider the 2-2-2 scenario, the above statements hold in any scenario.

Beyond quantum mechanics, this can also be understood with a combinatorial approach. For the 2-2-2 scenario, we have 16 possible events and the events we consider are only of the form $e = (a,b\vert x,y)$, \emph{i.e.}, it tells only about two measurement outcomes on the four possible measurements. Considering the possible other outcomes (even if we don't have access to all outcomes simultaneously because of incompatible measurements) a given event can address 4 possible cases out of the 16 total events. Two exclusive events do not overlap (if they do, they are not exclusive by definition). Hence, 4 mutually exclusive events cover all the 16 cases and the sum of their probabilities saturate the exclusivity principle.
We use these properties to visualize the different saturation of the exclusivity principle we use to build paradoxes.

\subsection{Saturation of the Exclusivity Principle in the Hardy Paradox}\label{sec:SEPHP}

The table is organized into four blocks. The first block has the list of all events, the second block and third blocks have the probabilities for the two exclusivity relations in \cref{eq:HP_EP}. The last block corresponds to the probability which is equal to zero under classical theory. The "N" (no) denotes an event that is forbidden (imposed by the conditions of the paradox), the "Y" (yes) denotes that an event is still available and correspond to one of the event in a given exclusivity principle equation. One can see that for both second and third block, all the events are covered, hence the exclusivity principle is saturated.

\begin{table}[ht]
\begin{tabular}{|c||c|c|c|c|c|c|c|c|c|c|c|c|c|c|c|c|}
	\hline
	$A_0$ & \multicolumn{8}{c|}{0} & \multicolumn{8}{c|}{1} \\
	\hline
	$A_1$ & \multicolumn{4}{c|}{0} & \multicolumn{4}{c|}{1} & \multicolumn{4}{c|}{0} &  \multicolumn{4}{c|}{1}  \\
	\hline
	$B_0$ & \multicolumn{2}{c|}{0} & \multicolumn{2}{c|}{1} & \multicolumn{2}{c|}{0} & \multicolumn{2}{c|}{1} & \multicolumn{2}{c|}{0} & \multicolumn{2}{c|}{1} & \multicolumn{2}{c|}{0} & \multicolumn{2}{c|}{1} \\
	\hline
	$B_1$ & 0 & 1 & 0 & 1 & 0 & 1 & 0 & 1 & 0 & 1 & 0 & 1 & 0 & 1 & 0 & 1 \\
	\hline
	\hline
	$P(0,0|A_0,B_1)$ & N &  & N &  & N &  & N & \multicolumn{9}{c|}{} \\
	\hline
	$P(1,1|A_1,B_1)$ & \multicolumn{5}{c|}{} & N & & N & \multicolumn{5}{c|}{}  & N & & N\\
	\hline
	$P(1,0|A_0,B_1)$ & \multicolumn{8}{c|}{} & Y & & Y & & Y & & Y & \\
	\hline
	$P(0,1|A_1,B_1)$ & & Y & & Y & \multicolumn{5}{c|}{} & Y & & Y & \multicolumn{4}{c|}{} \\
	\hline
	\hline
	$P(0,0|A_1,B_0)$ & \multicolumn{2}{c|}{N} & \multicolumn{6}{c|}{} & \multicolumn{2}{c|}{N} & \multicolumn{6}{c|}{} \\
	\hline
	$P(1,1|A_1,B_1)$ & \multicolumn{5}{c|}{} & N & & N & \multicolumn{5}{c|}{}  & N & & N\\
	\hline
	$P(0,1|A_1,B_0)$ &  \multicolumn{2}{c|}{} & \multicolumn{2}{c|}{Y} & \multicolumn{6}{c|}{} & \multicolumn{2}{c|}{Y} & \multicolumn{4}{c|}{} \\
	\hline
	$P(1,0|A_1,B_1)$ &  \multicolumn{4}{c|}{} & Y & & Y & \multicolumn{5}{c|}{} & Y & & Y &\\
	\hline
      \hline
      $P(0,0|A_0,B_0)$ & \multicolumn{2}{c|}{?} & \multicolumn{2}{c|}{} & \multicolumn{2}{c|}{?} & \multicolumn{10}{c|}{} \\
      \hline
\end{tabular}
	\caption{Table representing the list of all combinations of events and the saturation of the exclusivity principle for the Hardy paradox.}
	\label{tab:SEPHP}
\end{table}

\subsection{Saturation of the Exclusivity Principle in the CHSH Paradox}\label{sec:SEPCHSH}

The table is organized into four blocks. The first block has the list of all events, the second, the third and the fourth blocks have the probabilities for the three exclusivity relations in Fig.~\ref{fig:CHSH_exclu}. The last block corresponds to the probabilities which is equal to zero under classical theory. The "N" (no) denotes an event that is forbidden (imposed by the conditions of the paradox), the "Y" denotes that an event is still available and correspond to one of the event in a given exclusivity principle equation. One can see that for both second and third block, all the events are covered, hence the exclusivity principle is saturated.

\begin{table}[ht]
\begin{tabular}{|c||c|c|c|c|c|c|c|c|c|c|c|c|c|c|c|c|}
      \hline
      $A_0$ & \multicolumn{8}{c|}{0} & \multicolumn{8}{c|}{1} \\
      \hline
      $A_1$ & \multicolumn{4}{c|}{0} & \multicolumn{4}{c|}{1} & \multicolumn{4}{c|}{0} &  \multicolumn{4}{c|}{1}  \\
      \hline
      $B_0$ & \multicolumn{2}{c|}{0} & \multicolumn{2}{c|}{1} & \multicolumn{2}{c|}{0} & \multicolumn{2}{c|}{1} & \multicolumn{2}{c|}{0} & \multicolumn{2}{c|}{1} & \multicolumn{2}{c|}{0} & \multicolumn{2}{c|}{1} \\
      \hline
      $B_1$ & 0 & 1 & 0 & 1 & 0 & 1 & 0 & 1 & 0 & 1 & 0 & 1 & 0 & 1 & 0 & 1 \\
      \hline
      \hline
      $P(1,1|A_0,B_0)$ & \multicolumn{10}{c|}{} & \multicolumn{2}{c|}{N} & \multicolumn{2}{c|}{} & \multicolumn{2}{c|}{N}\\
      \hline
      $P(0,0|A_0,B_1)$ & N & & N & & N & & N & \multicolumn{9}{c|}{}\\
      \hline
      $P(1,0|A_0,B_0)$ & \multicolumn{8}{c|}{} & \multicolumn{2}{c|}{Y} & \multicolumn{2}{c|}{} & \multicolumn{2}{c|}{Y} & \multicolumn{2}{c|}{}\\
      \hline
      $P(0,1|A_0,B_1)$ & & Y &  & Y &  & Y &  & Y &\multicolumn{8}{c|}{} \\
      \hline
      \hline
      $P(0,0|A_0,B_1)$ & N & & N & & N & & N & \multicolumn{9}{c|}{}\\
      \hline
      $P(0,1|A_1,B_1)$ & & N & & N & \multicolumn{5}{c|}{} & N & & N & \multicolumn{4}{c|}{}\\
      \hline
      $P(1,0|A_0,B_1)$ & \multicolumn{8}{c|}{} & Y & & Y & & Y & & Y & \\
      \hline
      $P(1,1|A_1,B_1)$ & \multicolumn{5}{c|}{} & Y & & Y & \multicolumn{5}{c|}{}  & Y & & Y\\
      \hline
      \hline
      $P(1,1|A_1,B_0)$ & \multicolumn{6}{c|}{} & \multicolumn{2}{c|}{N} & \multicolumn{6}{c|}{} & \multicolumn{2}{c|}{N} \\
      \hline
      $P(0,1|A_1,B_1)$ & & N & & N & \multicolumn{5}{c|}{} & N & & N & \multicolumn{4}{c|}{}\\
      \hline
      $P(0,0|A_1,B_1)$ & Y & & Y & \multicolumn{5}{c|}{} & Y & & Y & \multicolumn{5}{c|}{}\\
      \hline
      $P(1,0|A_1,B_0)$ & \multicolumn{4}{c|}{} & \multicolumn{2}{c|}{Y} & \multicolumn{6}{c|}{} & \multicolumn{2}{c|}{Y} & \multicolumn{2}{c|}{} \\
      \hline
      \hline
      $P(0,1|A_0,B_0)$ & \multicolumn{6}{c|}{} & \multicolumn{2}{c|}{?} & \multicolumn{6}{c|}{} & \multicolumn{2}{c|}{?} \\
      \hline
      $P(0,1|A_1,B_0)$ &  \multicolumn{2}{c|}{} & \multicolumn{2}{c|}{?} & \multicolumn{6}{c|}{} & \multicolumn{2}{c|}{?} & \multicolumn{4}{c|}{} \\
      \hline
\end{tabular}
		\caption{Table representing the list of all combinations of events and the saturation of the exclusivity principle for the CHSH paradox.}
		\label{tab:SEPCHSH}
\end{table}

\section{CHSH Paradox with Classical Resources}\label{sec:CHSHC}

A deterministic theory is a classical theory where the probability of an event is either $0$ or $1$. A convex combination of from all deterministic probabilities can generate any classical behavior. Moreover in a geometrical approach classical theory forms a convex polytope where deterministic probability distribution are the vertices, making these probability distribution sufficient to consider when we wish to test whether if a given probability distribution is inside or outside this convex polytope \cite{BCPSW:rmp14}. In the graph framework, this corresponds to assign probabilities equal to zero or one to the vertices. However, because connected vertices represent exclusive events their assigned probabilities cannot be $1$ to both. The only possible combinations are $0-0$, $1-0$ and $0-1$ for a pair of connected vertices.

\begin{theorem}
A non-contextual classical theory cannot verify the CHSH Paradox.
\end{theorem}

\begin{proof}
From the previous work in \cite{CBCB:prl13,SZDM:pra16} and our investigation in this letter we know that the pentagon cannot be verified by any classical theory.

Any probability distribution of measurement outcomes defined by any classical theory can be written by as a convex combination of deterministic probability distribution \cite{BCPSW:rmp14}. Deterministic probability distributions are

\begin{equation}
  p_{det}(a,b\vert x,y) =
  \left\{
      \begin{aligned}
        &1 \text{ for $a=a_x$ and $b=b_y$},\notag \\
        &0 \text{ otherwise}\\
      \end{aligned}
    \right.
\end{equation}

where $a_x$ and $b_y$ $\forall x,y$ define a deterministic probability distribution. Hence, a classical probability distribution is

\begin{equation}
  p_{classical}(a,b\vert x,y) = \sum_{\lambda} q_\lambda p_{det}(a,b\vert x,y,\lambda)
\end{equation}

with $q_\lambda \geq 0$ and $\sum_{\lambda} q_\lambda = 1$. For a classical probability distribution to verify the paradox it is needed that all deterministic probability distribution verify

\begin{align}
& P(11|00)=0,\label{eq:CHSH_paradox_strict1}\\
& P(00|01)=0,\label{eq:CHSH_paradox_strict2}\\
& P(11|10)=0,\label{eq:CHSH_paradox_strict3}\\
& P(01|11)=0.\label{eq:CHSH_paradox_strict4}
\end{align}

Then, either there is at least one deterministic probability distribution such that $P(01|00)+P(01|10)>0$ or $P(01|00)>0$ or $P(01|10)>0$. None of these cases is actually possible. This can be verified directly on the graph. In Fig.~\ref{fig:CHSH_C} is showed the conditions of the CHSH paradox. The red rectangle defines a set where a deterministic probability distribution has to assign an event with probability equal to one inside, this correspond to the conditions in \cref{eq:CHSH_paradox_strict1,eq:CHSH_paradox_strict2,eq:CHSH_paradox_strict3,eq:CHSH_paradox_strict4}. While doing this, to verify the CHSH paradox at least one of the hyper-graph defined by the blue squares needs to have a probability equal to one assigned. From the graph this can be easily seen by trying to assign a probability equal to one to any of the hyper-graphs defined by the blue squares. It is impossible to have a probability equal to one assigned to each hyper-graph defined by the red rectangles.

\begin{figure}[ht]
      \centering
      \includegraphics[width=6cm]{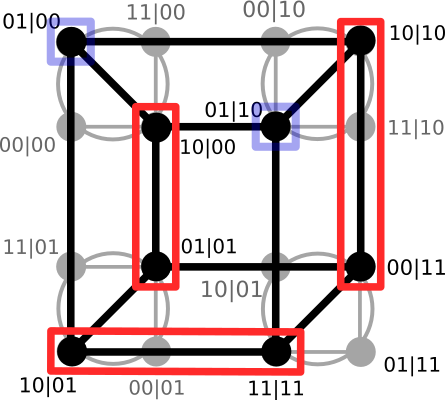}
      \caption{Exclusivity graph of the CHSH inequality in the form given in \cite{CSW:prl14}. The conditions of the CHSH paradox are represented by hyper-graphs defined by red rectangles and blue squares in the figure.}
      \label{fig:CHSH_C}
\end{figure}
\end{proof}

\section{CHSH Paradox with a Two-qubit State and Local Measurements}\label{sec:CHSHQ2D2D}
We now address the verification of the CHSH paradox introduced in \cref{eq:CHSH_paradox1,eq:CHSH_paradox2,eq:CHSH_paradox3,eq:CHSH_paradox4,eq:CHSH_paradox5}  with a quantum resource composed by a two-qubit state and local measurements.

\begin{theorem}
Given two parties, Alice and Bob, sharing a two-qubit entangled state $\vert \psi \rangle_{AB}$. There exist no possible set of local measurements $A_1$ and $A_2$ (resp. $B_1$ and $B_2$) for Alice's dichotomic measurements on her qubit (resp. Bob) that verify the CHSH paradox.
\end{theorem}

\begin{proof}
We show that in order to verify some of the conditions in the CHSH paradox imposes either to have a non-entangled state or compatible local measurements on one party side. Both are known to be required resource for non-locality. Hence if some of the conditions of the paradox imposes this, there would be no verification of the paradox. Let's consider the following generic two-qubit state:

\begin{equation}
	\vert \psi \rangle = a \vert 00 \rangle + e^{i \phi_b} b \vert 01 \rangle + e^{i \phi_c} c \vert 10 \rangle+ e^{i \phi_d} d \vert 11 \rangle,
\end{equation}

where $\vert a \vert^2 + \vert b \vert^2 + \vert c \vert^2 + \vert d \vert^2 = 1$ and $\phi_b,\phi_c,\phi_d \in [0,2\pi]$. One can show that the paradox cannot be verified from the two following conditions from \cref{eq:CHSH_paradox2,eq:CHSH_paradox4}

\begin{align}
& P(11|00)=0,\\
& P(11|10)=0.
\end{align}

In the CHSH scenario the two parties (Alice and Bob) have two dichotomic measurements denoted $A_0$ and $A_1$ for Alice and $B_0$ and $B_1$ for Bob. In particular, for a two-qubit state, each outcome of a measurement can be associated to a rank one projector. Without loss of generality, we can use the following eigenvectors:

\begin{align}
& \vert v_{A0} \rangle = \vert 0 \rangle, \\
& \vert v_{A1} \rangle = \cos(\frac{\theta_{A1}}{2})\vert 0 \rangle + e^{i \phi_{A1}} \sin(\frac{\theta_{A1}}{2})\vert 1 \rangle,\\
& \vert v_{B0} \rangle = \vert 0 \rangle,\\
& \vert v_{B1} \rangle = \cos(\frac{\theta_{B1}}{2})\vert 0 \rangle + e^{i \phi_{B1}} \sin(\frac{\theta_{B1}}{2})\vert 1 \rangle.
\end{align}

Then, the equation $P(11|00)=0$ imposes

\begin{equation}
	\vert \langle \psi \vert 11 \rangle \vert ^2 = 0,
\end{equation}

in other words, $d = 0$. Hence the state is

\begin{equation}
	\vert \psi \rangle = a \vert 00 \rangle + e^{i \phi_b} b \vert 01 \rangle + e^{i \phi_c} c \vert 10 \rangle,
\end{equation}

where $\vert a \vert^2 + \vert b \vert^2 + \vert c \vert^2 = 1$ and $\phi_b,\phi_c\in [0,2\pi]$.

Then, the equation $P(11|10)=0$ imposes

\begin{align}
	& \vert \langle \psi \vert (\sin(\frac{\theta_{A1}}{2})\vert 01 \rangle + e^{i \phi_{A1}} \cos(\frac{\theta_{A1}}{2})\vert 11 \rangle)\rangle \vert^2 = 0,\\
	& \sin(\frac{\theta_{A1}}{2})^2\vert \langle \psi \vert 01 \rangle\vert^2 = 0.
\end{align}

To satisfy this condition, there are two possibilities. If $\theta_{A1} = 0$ it implies $\vert v_{A0} \rangle = \vert v_{A1} \rangle$ which would not bring any non-local features in the measurements statistics. If $\vert \langle \psi \vert 01 \rangle\vert^2 = 0$ it implies $b=0$, hence the state would be non entangled. In both cases we can see that we cannot construct necessary conditions to observe non-locality.
\end{proof}

\section{New Contextual Inequality}

We review the graph theory approach of quantum correlations developed in \cite{CSW:arXiv10}. We define a graph $G(V,E)$, for which we associate to each vertex $i\in V$ an \textit{atomic} event, $(a\vert x)$, of a particular dichotomic measurement where the outcome $a=0$ or $a=1$ and the measurement is denoted by $x$. The edges represent the exclusivity and the compatibility of the measurements. Measurements are compatible if it is possible to perform them simultaneously. dichotomic measurements are exclusives if they cannot both have an outcome $1$, \emph{i.e.}, it is not possible that exclusive measurements have the outcome $1$ simultaneously.

Thus for all pairs of adjacent vertices, $(i,j) \in E$, the probability to have the measurement outcome $1$ assigned to both vertices is:
\begin{equation}\label{eq:hp1}
P(1,1\vert i,j)=0,
\end{equation}
where $p(a,b|c,d)$ represents the probability of getting results $a,b$ given measurement settings $c,d$.

We now consider how different classes of physical theories can assign outcomes on a given graph.

\subsection{Classical Theory}

In a non-contextual deterministic behavior, each dichotomic measurement leads to a predefined outcome $0$ or $1$ independently to the measurement's context. In this way each vertex $i\in V$ of the graph $G(V,E)$ has an assigned measurement outcome. Hence the probability that a particular event occurs will be equal to $0$ or $1$. As explained before, because the edges of the graph represent the exclusivity of the measurements, two adjacent vertices on the graph cannot simultaneously have the outcome $1$ assigned.

\subsection{Quantum Theory}

In quantum physics, the dichotomic measurements we use are represented by rank one projectors $\{\Pi_i=\vert v_i \rangle \langle v_i \vert\}$ with the normalized eigenvectors $\vert v_i \rangle$, where outcome $1$ is associated to projector $\Pi_i$ and outcome $0$ is associated to $\mathbb{I}-P_i$. This is equivalent to associating a unit vector $\vert v_i \rangle$ to the vertex $i$. In this framework, the exclusivity and compatibility relations between two measurements correspond to an orthogonality relation between the two unit vectors of the adjacent vertices. In graph theory, this corresponds to the orthonormal representation of a $\bar{G}$ of the graph.

\subsection{Contextuality Inequality}

The graph that we defined previously can be used to derive non-contextuality inequalities as follows.

For classically assigned outcome, we then arrive at the following inequality
\begin{equation}\label{eq:generalkcbs}
\sum_{i\in V}  P(1\vert i)\leq \alpha(G),
\end{equation}
where $\alpha(G)$ is the independence number of the graph $G$ (\emph{i.e.}, the maximum number of vertices that are not connected to each other).

It is sufficient to consider classical behavior to be restricted to the deterministic behaviors. Any classical behavior is obtained by a convex combination of deterministic behaviors. As a consequence, the maximum value possible in \cref{eq:generalkcbs} is obtained by a deterministic behavior. Hence, we have $P(1\vert i)= 0$ or $1$. To obtain the value of $\alpha(G)$, we can consider a coloring problem where the assigned $0/1$-probability values correspond to the two different colors. This can be seen as the maximum number of events with an outcome equals to $1$ simultaneously.

The maximum value of $\sum_{i\in V}  P(1\vert i)$ for a quantum resource is given by $\sum_{i\in V}  max \vert \langle v_i \vert \psi \rangle\vert ^2$. This corresponds to a well studied quantity in graph theory and is called Lovasz number or theta Lovasz where the eigenvectors with eigenvalue $1$ of the rank one projectors correspond to the so called orthonormal representation of the complement $\bar{G}$ of the graph $G$. It is also possible to derive the Lovasz number using Semi-Definite-Programming (SDP) \cite{Lovasz:79}.

\begin{align}\label{eq:beta}
\beta = \sum_{i\in V} P(1\vert i) \leq \alpha(G)\leq \vartheta(G)=\sum_{i\in V}  max \vert \langle v_i \vert \psi \rangle\vert ^2.
\end{align}

\subsection{CHSH Contextuality Inequality}\label{sec:CHSHCSW10}

We can use the exclusivity graph of the CHSH inequality to build a contextuality inequality for a different scenario. Instead of assigning events of two outcomes of two compatible measurements to the vertices of the graph we assign atomic events. Hence, the number of measurements is now equal to 8 and we call $\{A_i\}$ the set of observables with $\pm 1$ outcomes assigned to the vertices $\{i\}$. The graph with the new labels is represented in Fig.~\ref{fig:CHSH_contineq}.

\begin{figure}[ht]
      \centering
      \includegraphics[width=5.2cm]{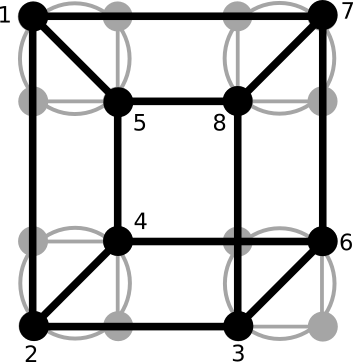}
      \caption{Relabeling of the vertices of the CHSH graph.}
      \label{fig:CHSH_contineq}
\end{figure}

We present our inequality in the following theorem:

\begin{theorem}
Given the exclusivity graph of the CHSH inequality, $G_\text{CHSH}(V_\text{CHSH},E_\text{CHSH})$ and a set of dichotomic measurements $\{A_i\}$ with $\pm 1$ outcomes values assigned to the vertices $\{i\}$, such that $\forall (i,j) \in E$ $A_i$ and $A_j$ are commuting. Any Non-Contextual Hidden Variable (NCHV) theory is bounded by the following inequality:
\begin{align}\label{eq:newcontineq}
\sum_{(i,j)\in E_{\rm CHSH}} \langle A_i A_j \rangle  \underset{{\rm NCHV}}{\geq} -6,
\end{align}
\end{theorem}

\begin{proof}
The theorem can be shown by expressing the NCHV bound by the graph properties. In this precise case, it is the same bound as in the CHSH inequality (as a consequence of using the same exclusivity graph). To this end we first express the expectation value in terms of the probabilities as following (for consistency with the rest of the letter we denote by "0" the outcome "1" and by "1" the outcome "-1")
\begin{align}
\langle A_i A_j \rangle  & = p(0,0\vert i,j) + p(1,1\vert i,j) - p(0,1\vert i,j) - p(1,0\vert i,j),\\
\langle A_i A_j \rangle  & = 1 - 2[p(0,1\vert i,j) + p(1,0\vert i,j)],
\end{align}
using the marginalization $p(1\vert j) = p(0,1\vert i,j) + p(1,1\vert i,j)$ and the exclusivity principle  $p(1,1\vert i,j) = 0$
\begin{align}
\langle A_i A_j \rangle  & = 1 - 2[p(1\vert j) + p(1\vert i)],
\end{align}
By summing over the all the compatible measurements we have
\begin{align}
\sum_{(i,j)\in E_{\rm CHSH}} \langle A_i A_j \rangle & =  12 - 6 \sum_{i\in V_{\rm CHSH}} p(1\vert i).
\end{align}
From \cite{CSW:arXiv10}, we have $\sum_{i\in V_{\rm CHSH}} p(1\vert i) \leq \alpha({G_\text{CHSH}}) = 3$. Hence
\begin{align}
\sum_{(i,j)\in E_{\rm CHSH}} \langle A_i A_j \rangle & \geq  12 - 6 ~\alpha({G_\text{CHSH}}),\\
\sum_{(i,j)\in E_{\rm CHSH}} \langle A_i A_j \rangle & \geq  - 6.
\end{align}
\end{proof}

In this scenario (see Fig.~\ref{fig:CHSH_contineq}) the CHSH Paradox takes the following form
\begin{align}\label{eq:CHSH_paradox_cont}
& P(1\vert 1)+P(1\vert 8)>0,\\
& P(1\vert 2)+P(1\vert 3)=1,\\
& P(1\vert 4)+P(1\vert 5)=1,\\
& P(1\vert 6)+P(1\vert 7)=1.
\end{align}
We finally can link the verification of the paradox and the violation of the inequality as follows

\begin{corollary}
The verification of the CHSH Paradox in the form of \cref{eq:CHSH_paradox_cont} implies the violation the following inequality:
\begin{align}\label{eq:newcontineq}
\sum_{(i,j)\in E_{\rm CHSH}} \langle A_i A_j \rangle  \underset{{\rm NCHV}}{\geq} -6.
\end{align}
\end{corollary}

\begin{proof}
The inequality can be rewritten as $\sum_{i\in V_{\rm CHSH}} p(1\vert i) \leq \alpha({G_\text{CHSH}}) = 3$. If the paradox in \cref{eq:CHSH_paradox_cont} is verified, the three last equations tells us that we have $\sum_{i\in V_{\rm CHSH}} p(1\vert i)= 3 + P(1\vert 1)+P(1\vert 8)$. Finally, because of the first equation one can conclude that the verification of the paradox implies $\sum_{i\in V_{\rm CHSH}} p(1\vert i) > 3$.
\end{proof}

\end{document}